\documentclass[conference]{IEEEtran}

\usepackage{cite}             
\usepackage[fleqn]{nccmath} 
\usepackage{amssymb} 
\usepackage{amsmath,bm}
\usepackage{amsthm} 
\usepackage{aligned-overset}
\usepackage[bookmarks=false]{hyperref}
\newtheoremstyle{mytheorem}
  {0pt}
  {0pt}
  {\itshape}
  {}
  {\bfseries}
  {:}
  {0.5em}
  {}

\theoremstyle{mytheorem}
\newtheorem{theorem}{Theorem}
\newtheorem{lemma}{Lemma}

\newtheorem{definition}{Definition}
\newtheorem{remark}{Remark}
\newcommand{\floor}[1]{\left \lfloor {#1} \right \rfloor }

\newcommand{\norm}[1]{\left\Vert#1\right\Vert}
\newcommand{\SNR}{\operatorname{SNR}}

\newcommand{\n}[1]{{\|\mathbf{#1}\|}} 

\newcommand{\R}{\mathbb{R}}

\newcommand{\cV}{{\cal V}}
\newcommand{\cC}{{\cal C}}

\def\b{\mathbf}
\NewDocumentCommand\sqn{mg}{%
    \|\mathbf{#1}_{\IfNoValueTF{#2}{}{#2}}\|^2%
}

\DeclareMathOperator{\EX}{\mathbb{E}}

\usepackage{graphicx}
\graphicspath{{figs/}}

\usepackage{tikz}
\usetikzlibrary{math} 
\usetikzlibrary{decorations.pathreplacing,calligraphy}
\usepackage{subfiles} 
\begin{document}

\title{PIR Over Wireless Channels: Achieving Privacy With Public Responses}

\author{%
 \IEEEauthorblockN{Or Elimelech and Asaf Cohen}
 \IEEEauthorblockA{\\The School of Electrical
and Computer Engineering\\
                   Ben-Gurion University of the Negev, Israel\\}
}\maketitle
\begin{abstract}
In this paper, we address the problem of Private Information Retrieval (PIR) over a public Additive White Gaussian Noise (AWGN) channel.
In such a setup, the server's responses are visible to other servers.
Thus, a curious server can listen to the other responses, compromising the user's privacy. 
Indeed, previous works on PIR over a shared medium assumed the servers cannot instantaneously listen to other responses.
To address this gap, we present a novel randomized lattice -- PIR coding scheme that jointly codes for privacy, channel noise, and curious servers which may listen to other responses.
We demonstrate that a positive PIR rate is achievable even in cases where the channel to the curious server is stronger than the channel to the user.
\end{abstract}
\begin{IEEEkeywords}
Private Information Retrieval, Wiretap, Lattice Codes, MISO.
\end{IEEEkeywords}
\section{Introduction}
In the classical PIR problem, there are $N$ identical, non-communicating servers, each storing the same $M$ messages, and a user who wishes to privately retrieve the $i$-th message without revealing the index $i$ to the servers. 
The primary objective is to minimize the communication overhead between the user and the servers while adhering to privacy constraints.

The PIR problem has recently attracted the attention of information theory society, motivated by the challenge of characterizing the fundamental limits of this problem.
The classical PIR capacity, with noiseless and orthogonal channels, is given by ${C_{PIR}=(1-1/N)(1-(1/N)^M)}$ in 
\cite{sun2017capacity}.
This key finding encourages the information theory community to explore further variations of this issue, tailored to the specific challenges arising naturally in real life.
Many extensions and generalizations of the problem have been examined.
For instance, robust PIR considers scenarios where some databases may fail to respond, and $T$-private PIR ensures privacy even when any $T$ out of the $N$ databases collude, as studied in \cite{sun2017capacityColluding}.
Another variation is PIR from byzantine databases, which accounts for the possibility of up to $B$ databases providing erroneous responses, either intentionally or unintentionally \cite{banawan2018capacity}.

While the PIR problem has been extensively studied over the past decade, most research has focused on simple communication channels characterized by \emph{orthogonal and noiseless links} between the user and the database. 
The case of PIR over noisy orthogonal channels (NPIR) was examined in \cite{banawan2019noisy}, which demonstrated that the channel coding required to mitigate channel errors is "almost separable" from the retrieval scheme, depending only on the agreed traffic ratio.
Additionally, \cite{banawan2019noisy} explored PIR over various types of MAC and showed that, unlike NPIR, the channel coding and retrieval
scheme is not optimal in general.
Building on this foundation, the authors in \cite{shmuel2021private,orelimcorrection24,orelimEfficientISIT} investigated PIR over a block-fading Gaussian MAC.

In these papers, the authors did not refer to the fact that the channel model, which represents a practical wireless channel, is publicly available, hence making it possible for a curious server to overhear other responses, which may cause a privacy leakage.
This paper aims to maintain privacy under the assumption that any of the servers can listen to the channel.
The main goal is to prevent the servers from gaining any information about the identity of the user's desired message from their observed signal. 

In \cite{wang2018capacity}, the author extends the classical PIR problem to address scenarios involving colluding servers and eavesdroppers. 
This extended problem introduces additional privacy and secrecy constraints to the conventional PIR setting.
Specifically, the user aims to retrieve a message without revealing the desired message index to any set of $T$ colluding, and an eavesdropper who can listen to the queries and answers of any $E$ servers but is prevented from learning
any information about the messages.
This related problem indeed adds a secrecy constraint, yet it assumes noiseless and orthogonal channels.
Our model assumes a Gaussian MAC channel between the servers and the user, which presents a significant difference.

To address these challenges, we propose a joint PIR-channel coding scheme that exploits the additive nature of the AWGN MAC while engineering the responses such that no server can infer the identity of the desired message.
Simultaneously, the responses, which "sum up over the air," enable the user to decode the desired message.
Additionally, the integration of wiretap coset coding ensures that any curious server is unable to decode the identity from the overheard responses, thereby preserving privacy.

Section \ref{sec:sys} outlines the system model and provides the necessary preliminaries.
Section \ref{sec:mainRes} presents the achievable PIR rates for different variations of the proposed model.
Section \ref{sec:Scheme} details the PIR achievable schemes and provides proofs for the Theorems introduced in Section \ref{sec:mainRes}.




\section{System Model Preliminaries}\label{sec:sys}

\subsection{Notational Conventions}
Throughout the paper, we use boldface lowercase to refer to vectors, e.g., $\b{h} \in \R^L$.  For a vector $\b{h}$, we write $\n{h}$ for its Euclidean norm, i.e. $\n{h} \triangleq \sqrt{\sum_{i}h_i^2}$. 
All logarithms in this paper are natural logarithms, and rates are in nats.

\subsection{System Model and Problem Statement}\label{sec-system model}

Consider the PIR problem over a public Gaussian wiretap channel with $N$ identical and non-communicating servers (transmitters), each equipped with one transmit antenna.
Each server stores a set of messages $W_1^M=\{W_1,W_2,...,W_M\}$ of size $L$ each. The messages were drawn uniformly and independently from $\mathbb{F}_p^L$ where $p$ is assumed to be prime
, i.e., in nats units, we have,
\begin{equation}\label{equ-Entropy of messages}
\begin{aligned}
&H(W_l)=L\log{p} \ \text{for} \  l=1,...,M,\\
&H(W_1^M)=ML\log{p}.
\end{aligned}
\end{equation}


As usual, the user wishes to privately retrieve the message $W_i$, where the index $i$ is assumed to be uniformly distributed on $[1,...,M]$, i.e., $i$ is a realization of $\theta \sim U[1,..M]$, while keeping $\theta$ secret from each server.
However, unlike traditional PIR setups, herein, since the replies are over a wireless channel, it is important to keep $\theta$ private, also considering the public replies.
Let $Q_j(\theta)$ denote a query that depends on the required random message index and is generated for the $j$th server.
$Q_j(i)$ will be used to show the query's dependence on the specific realization of the index, $i$.
Accordingly, the user generates a set of $N$ queries $Q_1(i),Q_2(i),...,Q_N(i)$, one for each server, which are statistically independent with the messages (as those are not known to him). 
That is, we have
\begin{equation}\label{eq:const_querry}
I(W_1^M;Q_1(\theta),...,Q_N(\theta))=0.
\end{equation}
\textcolor{black}{
We assume the servers receive the queries through independent control channels and cannot access each other’s queries or answers}\footnote{This assumption is a common assumption in the literature and is justified by the fact that queries are low-rate messages, allowing encryption to be implemented without incurring significant overhead.}

The $k$th server generates an answer based on the received query $Q_k(\theta)$ and the stored messages $W_1^M$.
The answer, denoted by $\b{x}_k(i)$, is of size $n$, where $n$ is fixed for all responses,\footnote{
This is because we follow the usual Gaussian MAC setup in the literature \cite{cover2012elements}, where a codeword is transmitted during $n$ channel uses.}.
The answer is a deterministic function of the messages and the query. 
We have
\begin{equation}\label{eq:const_deter}
H(\b{x}_k(\theta) | W_1^M,Q_k(\theta))=0.
\end{equation}

Because the server's responses are transmitted over a public channel, a byzantine server could monitor the communication within the channel, potentially compromising user privacy (Figure \ref{fig:PIR_Wiretap_simpleModel}).
To ensure privacy, each server should not be able to infer any information about the desired message index given the queries, the entire database, and the received signal $\mathbf{w}_j$ at the $j$-th server antenna.
Thus, for each server $j$ the following \textit{privacy} constraint must be satisfied,

\begin{equation}\label{equ-secrecy-privacy}
\begin{split}   
\hspace{-0.5cm}\text{[Privacy - under public responses]}&\\
&\hspace{-4.5cm}\lim_{n\rightarrow\infty}I(\theta;Q_j(\theta),W_1^M,\mathbf{w}_j)=0, \hspace{0.2cm}\forall j\in\{1,...,N\}.
\end{split}
\end{equation}
\noindent
Note that the privacy requirement \eqref{equ-secrecy-privacy}, as defined above, is a special case arising from all servers sharing the same communication channel.
The main difference lies in the curious server's knowledge of
$\mathbf{w}_j$, which represents the received signal at its antenna.
Additionally, it is important to highlight that our privacy requirement is asymptotic, as it relies on the strong secrecy notion to ensure no information leakage from $\mathbf{w}_j$.
    
The databases are linked to the user via a block-memoryless fading AWGN channel (Figure \ref{fig:PIR_Wiretap_simpleModel}).
In this setup, the channel remains constant throughout the transmission of codewords of size $n$, and each block is independent of the others.
Thus, over a transmission of $n$ symbols, the user and the $j$-th server observe a noisy linear combination of the transmitted signals,
\begin{equation}\label{eq:MAC_output_model}
\mathbf{y}=\sum_{k=1}^{N}h_k\mathbf{x}_k+\mathbf{n}_y, \quad \mathbf{w}_j=\sum_{k=1}^{N}g_{j,k}\mathbf{x}_k+\mathbf{n}_{j,w}
\end{equation}
Here, $h_k\thicksim\mathcal{N}(0,1)$ and $g_{j,k}\thicksim\mathcal{N}(0,1)$ represent the real channel coefficients of the servers-user and servers-server channel, respectively, and $\mathbf{n}_y\thicksim\mathcal{N}(0,\sigma_y^2\mathbf{I}^{n\times n})$, $\mathbf{n}_{j,w}\thicksim\mathcal{N}(0,\sigma_w^2\mathbf{I}^{n\times n})$ are i.i.d. Gaussian noise observed by the user and the curious server, respectively.

Let us refine certain aspects of our model assumptions.
We assume that explicit cooperation among servers is not permitted.
Additionally, we assume a per-database power constraint, where all transmitting databases operate with a fixed power $P$, and power cannot be allocated differently to different servers.
Thus, the transmitted codebook $\cC$ must satisfy the average power constraint, i.e., $\EX\left[\|\b{x}_k\|^2\right] \leq nP$. 

We assume there is full channel state information at the transmitter (CSIT) and together with statistical information of the curious server, i.e., $\{h_k\}_{k=1}^N$,$\{g_{j,k}\}_{k=1}^N$ for any $j$, and $\sigma_w^2$ are known.
\textcolor{black}{
We also assume reciprocity between the channels, meaning that the forward and reverse channels between the servers and the user exhibit the same properties and the same channel coefficients.}
\textcolor{black}{Under the CSIT assumption, the curious server can either cancel its own transmitted message or remain idle while passively listening.
Consequently, if the index of the malicious server is $j$, it obtains 
\begin{equation}\label{eq:CSIT_malicuse}
    \mathbf{z_j}=\mathbf{w}_j-g_{j,j}\mathbf{x}_j=\sum_{k=1, k\neq j}^{N}g_{j,k}\mathbf{x}_k+\mathbf{n}_{j,w}.
\end{equation}}

Upon receiving the mixed response $\mathbf{y}$ from all the databases, the user decodes the required message $W_i$. 
Let $\widehat{W_i}$ denote the decoded message at the user and define the error probability for decoding a message as follows,
\begin{equation}\label{equ-probability of error definition}
P_e(L)\triangleq P_r(\widehat{W_i} \neq W_i).
\end{equation}
We require that $P_e(L)\rightarrow 0$ as $L$ tends to infinity.

\begin{figure}
    \centering
    \includegraphics[width=0.7\linewidth]{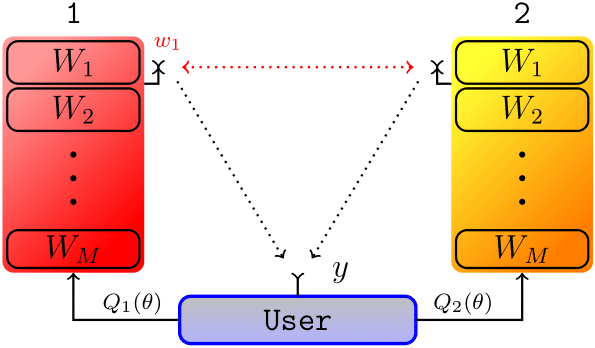}
    \caption{System model with two servers connected to a user via an AWGN MAC channel. 
    Server $1$ acts as a malicious entity, monitoring the communication and having access to $w_1$. }
    \label{fig:PIR_Wiretap_simpleModel}
\end{figure}

\subsection{Performance Metric}
For a joint PIR-coding scheme over MAC, we define the PIR rate as the total number of desired bits divided by the total number of channel uses \cite{shmuel2021private,orelimEfficientISIT}.
Specifically, 
\begin{definition} Denote:
 $$R_{PIR}^{MAC}(n)\triangleq \frac{H(W_i)}{n}=\frac{L\log{p}}{n},$$ 
 where $n$ represents the number of channel uses.
 The PIR capacity over AWGN MAC, denoted by $C_{PIR}^{MAC}$, is the supremum of $R_{PIR}^{MAC}(n)$, under which reliable communication is achievable, ensuring the privacy of the user's will. 
 That is, satisfying (\ref{eq:const_querry},\ref{eq:const_deter},\ref{equ-secrecy-privacy}) and (\ref{equ-probability of error definition}). 
\end{definition}

\textcolor{black}{
Neglecting the privacy constraint, our model matches with Gaussian MAC with per-antenna power constraint model \cite{cover2012elements}, as the servers cannot explicitly cooperate but may implicitly do so through the queries together with the assumption of full CSIT.
An upper bound for this model can be established by having all antennas transmit the same symbol simultaneously, effectively reducing the system to an AWGN channel with a single antenna and a total power equal to the sum of the individual powers. Consequently, this setup serves as an upper bound for the PIR capacity rate in our model.
\begin{equation}\label{upper bound}
C_{AWGN}^{MAC}=\frac{1}{2}\log\left( 1+ \frac{P_{total}}{\sigma^2}\right).
\end{equation}
}

The following subsection introduces some preliminaries on the lattice, which we will need to describe and analyze our scheme. 

\subsection{Lattice Code and Lattice Gaussian Distribution}
It has been shown in \cite{erez2004achieving,ling2014achieving} that the full capacity of the point-to-point AWGN channel is achievable using lattice encoding and decoding. 
We now provide a brief background on lattice codes, which will be useful in the remainder of this paper.

An $n$-dimensional lattice $\Lambda$ is a discrete subgroup of the Euclidean space $\mathbb{R}^n$, which is closed under reflection and real addition. Recall the definition of an $n$-dimensional lattice \cite{zamir2014lattice}.
\begin{definition}\label{Lattice}
A non-degenerate $n$-dimensional lattice  is defined by a set of $n$ linearly independent basis (column) vectors $\mathbf{g}_1,..., \mathbf{g}_n$ in $\mathbb{R}^n$.
The lattice  is composed of all integral combinations of the basis vectors, i.e.,
\begin{equation*}
    \Lambda=\left\{\lambda=\sum^{n}_{k=1} i_k\mathbf{g}_k \text{ : } i_k \in \mathbb{Z} \right\}
    =\left\{\lambda=G\cdot\mathbf{i} \text{ : } \mathbf{i} \in \mathbb{Z}^n \right\}
\end{equation*}
where $\mathbb{Z} = \{0, \pm 1, \pm 2,...\}$ is the set of integers, $\mathbf{i} = (i_1,...,i_n)^t$
is an $n$-dimensional integer (column) vector, and the $n \times n$ generator matrix $G$ is given
by
$$G=\left[\ \mathbf{g}_1\ |\ \mathbf{g}_2\ |\ ...\ |\ \mathbf{g}_n\ \right]$$
The resulting lattice is denoted $\Lambda(G)$.
\end{definition}

\begin{definition}[Quantizer]
A lattice quantizer is a map, ${Q_{\Lambda}: \mathbb{R}^n \rightarrow  \Lambda}$, that sends
 a point, $\mathbf{s}$, to the nearest lattice point in Euclidean distance. That is,
 \begin{equation}
     Q_{\Lambda}(\mathbf{s})=\text{argmin}_{\lambda\in\Lambda}\norm{\mathbf{s-\lambda}}.
 \end{equation}
\end{definition}
\begin{definition}[Voronoi Region]
The \textit{fundamental Voronoi region}, $\mathcal{V}$, of a lattice, $\Lambda$, is the set of all points in $\mathbb{R}^n$ that are closest to the zero vector compared to any other lattice point. That is, 
${\mathcal{V}=\{\mathbf{s}: Q_{\Lambda}(\mathbf{s})=0\}}$.
\end{definition}
\begin{definition}[Modulus] Let $[\mathbf{s}] \text{mod} \ \Lambda$ denote the quantization error of $\mathbf{s}\in\mathbb{R}^n$ with respect to the lattice $\Lambda$. That is,
\begin{equation}
[\mathbf{s}] \text{mod} \ \Lambda =\mathbf{s}-Q_{\Lambda}(\mathbf{s}).
\end{equation}   

For all $\mathbf{s,t} \in \mathbb{R}^n$ and $\Lambda_c \subseteq \Lambda_f$, the $\text{mod} \ \Lambda$ operation satisfies:
\begin{equation}\label{eq:dist}
[\mathbf{s} + \mathbf{t}] \text{mod} \ \Lambda =
\big[[\mathbf{s}] \text{mod} \ \Lambda +\mathbf{t}\big] \text{mod} \ \Lambda 
\end{equation}

\begin{equation}\label{eq:QuantizationMod}
\left[Q_{\Lambda_f}(\mathbf{s})\right]\ \text{mod} \ \Lambda_c =
\left[Q_{\Lambda_f}([\mathbf{s}]\ \text{mod} \ \Lambda_c)\right]\ \text{mod} \ \Lambda_c 
\end{equation}

\begin{equation}\label{eq:Zscal}
[a\mathbf{s}] \text{mod} \ \Lambda =
[a[\mathbf{s}] \text{mod} \ \Lambda]\text{mod} \ \Lambda \quad \forall a\in \mathbb{Z} 
\end{equation}

\begin{equation}\label{eq:scal}
\beta [\mathbf{s}] \text{mod} \ \Lambda =
[\beta \mathbf{s}] \text{mod} \ \beta \Lambda \quad \forall \beta\in \mathbb{R} 
\end{equation}

\end{definition}

Next, we recall the definitions from \cite{ling2014semantically} for Lattice Gaussian distributions, which will be needed to construct the PIR scheme.
For $\sigma>0$ and $\mathbf{c} \in \mathbb{R}^{n}$, we define the Gaussian distribution of variance $\sigma^{2}$ centered at $\mathbf{c} \in \mathbb{R}^{n}$ as

$$
f_{\sigma, \mathbf{c}}(\mathbf{x})=\frac{1}{(\sqrt{2 \pi} \sigma)^{n}} e^{-\frac{\|\mathbf{x}-\mathbf{c}\|^{2}}{2 \sigma^{2}}}
$$
for all $\mathbf{x} \in \mathbb{R}^{n}$. 
\begin{definition}[$\Lambda$-Periodic Function]
\begin{equation*}
f_{\sigma, \Lambda}(\mathbf{x})=\sum_{\boldsymbol{\lambda} \in \Lambda} f_{\sigma, \boldsymbol{\lambda}}(\mathbf{x})=\frac{1}{(\sqrt{2 \pi} \sigma)^{n}} \sum_{\boldsymbol{\lambda} \in \Lambda} e^{-\frac{\|\mathbf{x}-\boldsymbol{\lambda}\|^{2}}{2 \sigma^{2}}} 
\end{equation*}
for all $\mathbf{x} \in \mathbb{R}^{n}$. 
Note that $f_{\sigma, \Lambda}$ restricted to the quotient $\mathbb{R}^{n} / \Lambda$ \footnote{the set of all cosets which in our case equivalent to the Voronoi region} is a probability density.
\end{definition}
\begin{definition}[Discrete Gaussian]
discrete Gaussian distribution over $\Lambda$ centered at $\mathbf{c} \in \mathbb{R}^{n}$ and taking values over $\boldsymbol{\lambda} \in \Lambda$ defined as:

$$
D_{\Lambda, \sigma, \mathbf{c}}(\boldsymbol{\lambda})=\frac{f_{\sigma, \mathbf{c}}(\boldsymbol{\lambda})}{f_{\sigma, \Lambda}(\mathbf{c})}, \quad \forall \boldsymbol{\lambda} \in \Lambda
$$
\end{definition} 

Lastly, we define the discrete Gaussian distribution over a coset of $\Lambda$, i.e., the shifted lattice $\Lambda-\mathbf{c}$ :

$$
D_{\Lambda-\mathbf{c}, \sigma}(\boldsymbol{\lambda}-\mathbf{c})=\frac{f_{\sigma}(\boldsymbol{\lambda}-\mathbf{c})}{f_{\sigma, \Lambda}(\mathbf{c})} \quad \forall \boldsymbol{\lambda} \in \Lambda .
$$

\section{Main results}\label{sec:mainRes}
The following theorem presents an achievable PIR rate for the AWGN channel with public responses and without fading, i.e., all the fading coefficients equal $1$. 
The first result is for $N=2$ servers.
\begin{theorem}\label{th:the PIR rate for 2S public AWGN MISO channel}
For the $N=2$ databases public AWGN channel, the following PIR rate is achievable,
\begin{equation}\label{eq- retrieval rate for non-fading MAC N=2}
R=\frac{1}{2} \log \left(\min \left\{\frac{\frac{1}{2}+SNR_y}{1+SNR_w}, SNR_{y}\frac{\frac{1}{2}+ SNR_y}{1+ SNR_y}\right\}\right)-\frac{1}{2}
\end{equation}
where $SNR_y=\frac{P}{\sigma^2_y}$, and $SNR_w=\frac{P}{\sigma^2_w}$.
\end{theorem}
The following Theorem extends the above result for $N>2$,
\begin{theorem}\label{th:the PIR rate for N>2 public AWGN MISO channel}
Consider the PIR problem with $N\geq2$ databases over a public AWGN channel. 
Then, for any non-empty subsets of databases $\mathcal{S}_1, \mathcal{S}_2$ satisfying ${\mathcal{S}_1\cap \mathcal{S}_2=\emptyset}$,
$\mathcal{S}_1 \cup \mathcal{S}_2 \subseteq\{1, ..., N\}$, the following PIR rate is achievable,
\begin{equation}\label{eq- retrieval rate for non-fading MAC N>2}
    R=\frac{1}{2} \log \left(\min \left\{\frac{\frac{1}{2}+SNR_y'}{1+SNR_w'}, SNR_{y}'\frac{\frac{1}{2}+ SNR_y'}{1+ SNR_y'}\right\}\right)-\frac{1}{2}
\end{equation}
where $SNR_y'=\floor{\frac{N}{2}}^{2}\frac{P}{\sigma^2_y}$, $SNR_w'=\floor{\frac{N}{2}}^{2}\frac{P}{\sigma^2_w}$.
\end{theorem}
Accordingly, the achievable rate scaling laws with the number of servers $N$ and the power $P$ are the same as those for the AWGN wiretap channel.
The next theorem considers the PIR over public AWGN block--fading channel,
\begin{theorem}\label{th:the PIR rate for block fading public AWGN MISO channel}
Consider the PIR problem with $N\geq2$ databases over a public block-fading AWGN. 
Then, for any non-empty subsets of databases $\mathcal{S}_1, \mathcal{S}_2$ satisfying ${\mathcal{S}_1\cap \mathcal{S}_2=\emptyset}$,
$\mathcal{S}_1 \cup \mathcal{S}_2 \subseteq\{1, ..., N\}$, the following PIR rate is achievable,
\begin{equation}\label{eq- retrieval rate for block-fading MAC}
R=\frac{1}{2} \log \left(\min_i \left\{\frac{\frac{1}{2}+SNR_{\tilde{y}}}{1+SNR_{\tilde{w}_i}}, SNR_{\tilde{y}}\frac{\frac{1}{2}+ SNR_{\tilde{y}}}{1+ SNR_{\tilde{y}}}\right\}\right)-\frac{1}{2}
\end{equation}
where  $\Tilde{h}_m = \sum_{k\in\mathcal{S}_m} |h_k|$, $\Tilde{g}_m = \sum_{k\in\mathcal{S}_m} g_k$, $SNR_{\tilde{y}} = \frac{\tilde{h}_1^2 P}{\sigma_{y}^2}$, 
$SNR_{\tilde{w}_1} = \frac{\tilde{g}_1^2P}{\sigma_{w}^2} $, and 
$SNR_{\tilde{w}_2} = \frac{\left(\tilde{g}_2\frac{\tilde{h}_1}{\tilde{h}_2}\right)^2 P}{\sigma_{w}^2} $.
\end{theorem}

\section{A Joint PIR Scheme Over Public AWGN MISO channel}\label{sec:Scheme}

This section presents the schemes and the proofs for a PIR over the public AWGN MISO scheme.

\subsection{A Joint Scheme for Secure PIR Over Non-Fading AWGN MAC With $N=2$ Servers}\label{sec:two_servers_nonfading}
Consider the simple case with $N=2$ servers, as depicted in Figure \ref{fig:PIR_Wiretap_simpleModel}.
To ensure privacy constraints, focusing on the server with the best channel conditions is sufficient. 
In this two-server scenario, channel reciprocity is assumed; thus, without loss of generality, we designate server $2$ as the malicious server.
Let $\b{y}$ and $\b{z}_2$ represent the received inputs at the user and the malicious server, respectively.
Over a transmission of $n$ symbols, the received input can be described as,
\begin{equation*}
    \b{y}(i)=\b{x}_1+\b{x}_2+\b{n}_y
\end{equation*}  
and
\begin{equation*}
    \b{z}_2(i)=\b{x}_1+\b{n}_{j,w}
\end{equation*}  
where $\b{n}_y\sim \mathcal{N}(0,\sigma^2_y\b{I}^{n\times n})$ and $\b{n}_w\sim \mathcal{N}(0,\sigma_w^2\b{I}^{n\times n})$  are the noises at the legitimate receiver and the curious server, respectively \eqref{eq:CSIT_malicuse},\eqref{eq:MAC_output_model}.
Note that the index $i$ in the received input denotes the private index of the desired message.

In the following, we provide an achievability scheme that proves Theorem \ref{th:the PIR rate for 2S public AWGN MISO channel}.
In our suggested scheme, we use the results provided in \cite{ling2014semantically}, which propose a wiretap lattice code that achieves strong secrecy and semantic security over the Gaussian wiretap channel (GWC).
Their scheme achieves the secrecy capacity of the GWC up to a constant gap of $\frac{1}{2}$ nat.
\subsubsection{Coding Scheme}
Let $\Lambda_w\subset \Lambda_f $ be a nested chain of $n$-dimensional lattices, such that
\begin{equation*}
    R=\frac{1}{n}\log |\Lambda_f/\Lambda_w|
\end{equation*}

\noindent
According to \cite[Lemma 5]{nazer2011compute}, there exists a one-to-one mapping function, $\phi(\cdot)$, between a element $\b{s}\in\mathbb{F}_p$, and a coset lattice point $\Tilde{\boldsymbol{\lambda}}\in \Lambda_f/\Lambda_w$. 
For simplicity, we assume the coset representatives set is ${\Lambda_f\cap\cV(\Lambda_w)}$, namely:
$$s\in\mathbb{F}_p \mapsto \boldsymbol{\lambda}=(\lambda_1,...,\lambda_n)\in\Lambda_f\cap\cV(\Lambda_w).$$

The set of coset representatives is not unique: One could choose $\boldsymbol{\lambda} \in \Lambda_f \cap \mathcal{R}(\Lambda_w)$ for any fundamental region $\mathcal{R}(\Lambda_w)$, not
necessarily the Voronoi region $\mathcal{V}(\Lambda_w)$.

\noindent \textit{Query:} The user, which is interested in the message $W_i$, generates a random vector $\b{b}$ of length $M$ such that each entry is either $1$ or $0$, independently and with equal probability. Then, the user sends the following queries to the servers,
\begin{equation}\label{equ-queries structure PIR scheme}
\begin{aligned}
&Q_{1}(i)=\b{b}, \ \ Q_{2}(i)=-\b{b}-\b{e}_i, \ \text{ if } b_i=0\\
&Q_{1}(i)=\b{b}, \ \ Q_{2}(i)=-\b{b}+\b{e}_i, \ \text{ if } b_i=1
\end{aligned}
\end{equation}

\noindent 
\textit{Answers:} Upon receiving the queries, the servers construct their responses by computing linear combinations of the messages, where the combining coefficients are determined by the entries in $Q_j(i)$. 
That is,
\begin{equation}\label{equ-servers answer formation}
\begin{aligned}
&\b{A}_k=\sum_{m=1}^M Q_{k}(i)W_m,
\end{aligned}
\end{equation}
where $Q_{k}(m)$ is the $m$th entry of the vector $Q_{k}$.
We note that $A_k=(a_k^{(1)},\dots,a_k^{(L)})\in \mathbb{F}_p^L$, and the scheme is focused on the transmission of a single symbol $a_k^{(m)}$, where $1\leq m \leq L$, from each answer.
To construct the entire message, the servers must iterate this process across all $L$ symbols.

Without loss of generality, assume $m=1$, i.e., the servers wish to transmit the first symbol from each answer, that is, $a_1^{(1)} $ and $a_2^{(1)}$.
Note that $\b{A}_1+\b{A}_{2}$ is equal to either $W_i$, or $-W_i$. 
This depends on the sign of $b_i$, which is known to the user.
In the same way $a_1^{(m)}+a_2^{(m)} = \pm W_i^{(m)}$.
To encode the symbols, each server maps its symbols to the codebook as follows: $\boldsymbol{\lambda}_1 = \boldsymbol{\phi}(a_1^{(1)})$ and  $\boldsymbol{\lambda}_2 = \boldsymbol{\phi}(a_2^{(1)})$.
Then, each server samples randomly and independently a lattice point from the relevant coset according to $D_{\Lambda_w+\boldsymbol{\lambda}_i,\sqrt{P}}$ where $i\in\{1,2\}$ and $P$ is the transmission power.
Equivalently, the servers transmissions are:
\begin{equation}
\begin{split}
    &\mathbf{x_1} = \boldsymbol{\lambda}_1 + \mathbf{r}_1, \ \ \mathbf{r}_1\sim D_{\Lambda_w,\sqrt{P},-\boldsymbol{\lambda}_1}
    \\&
    \mathbf{x_2} = \boldsymbol{\lambda}_2 + \mathbf{r}_2, \ \ \mathbf{r}_2\sim D_{\Lambda_w,\sqrt{P},-\boldsymbol{\lambda}_2}
\end{split}
\end{equation}
\textcolor{black}{This construction introduces uncertainty for the curious server, preventing it from inferring the desired message.}

\textit{Decode:} To decode $\mathbf{v}\overset{\Delta}{=}\boldsymbol{\phi}(W_i^{(1)})$, the user computes the following, 
\begin{equation*}
\hat{\b{v}}=[Q_{\Lambda_f}(\alpha\b{y})]\ \text{ mod } \Lambda_w,
\end{equation*}   
where $\alpha=\frac{P}{P+\sigma^2_y}$ is the MMSE coefficient \cite{ling2014semantically}.
\textcolor{black}{Note that the modulo operation is with respect to $\Lambda_w$.}
The above reduces to the Modulo-Lattice Additive Noise (MLAN) channel (\cite{erez2004achieving}) as follows,
\begin{equation*}
    \begin{split}
        \hat{\b{v}}&=[Q_{\Lambda_f}(\alpha\b{y})]\ \text{ mod } \Lambda_w
        \\&
        =[Q_{\Lambda_f}(\alpha(\b{x}_1+\b{x}_2+\b{n}_y))]\ \text{ mod } \Lambda_w
        \\&
        =[Q_{\Lambda_f}(\b{x}_1+\b{x}_2+(\alpha-1)(\b{x}_1+\b{x}_2)+\alpha\b{n}_y))]\ \text{ mod } \Lambda_w
        \\&
        =[Q_{\Lambda_f}(\b{x}_1+\b{x}_2+\b{\Tilde{n}}_y))]\ \text{ mod } \Lambda_w
    \end{split}
\end{equation*}
\noindent where we denote the equivalent noise as ${\b{\Tilde{n}}_y\overset{\Delta}{=}(\alpha-1)(\b{x}_1+\b{x}_2)+\alpha\b{n}_y}$ where $\frac{1}{n}\EX[\norm{\b{\Tilde{n}}_y}]=0$. 
Since $\mathbf{x}_{1}$ and $\mathbf{x}_{2}$ drawn independently  from $D_{\Lambda_w+\boldsymbol{\lambda}_i,\sqrt{P}}$, the second moment of ${\b{\Tilde{n}}_y}$ is given by
${\Tilde{\sigma}_y^{2}=\frac{1}{n}\EX\left[\norm{\mathbf{n}_{eq}^2}\right] =2(1-\alpha)^2P+\alpha^2 \sigma_{y}^2}$ where we can optimize it on $\alpha$ getting $\alpha_{opt}=\frac{2P}{2P+\sigma_{y}^2}$
and the resulting optimal second moment
is $\Tilde{\sigma}^{2}_{y,opt}=\frac{2P\sigma_{y}^2}{2P+\sigma_{y}^2}$.

Using \eqref{eq:QuantizationMod} we have,
\begin{equation*}
    \begin{split}
        &[Q_{\Lambda_f}(\b{x}_1+\b{x}_2+\b{\Tilde{n}}_y)]\ \text{ mod } \Lambda_w
        \\
       &  = [Q_{\Lambda_f}([\b{x}_1+\b{x}_2+\b{\Tilde{n}}_y]\text{ mod } \Lambda_w)]\ \text{ mod } \Lambda_w
       \\
       &   = [Q_{\Lambda_f}([\b{v}+\b{r}_1+\b{r}_2+\b{\Tilde{n}}_y]\text{ mod } \Lambda_w)]\ \text{ mod } \Lambda_w
       \\
       &\overset{(a)}{=} [Q_{\Lambda_f}([\b{v}+\b{\Tilde{n}}_y]\text{ mod } \Lambda_w)]\ \text{ mod } \Lambda_w
    \end{split}
\end{equation*}
\noindent (a) since $\b{r}_1\in\Lambda_w$ and $\b{r}_2\in\Lambda_w$ together with \eqref{eq:dist}.
Therefore
$$\hat{\b{v}}=[Q_{\Lambda_f}(\b{v}+\b{\Tilde{n}}_y)]\ \text{ mod } \Lambda_w$$


Then, from the \textit{AWGN-goodness} of $\Lambda_f$, it follows that $P_e$ tends to $0$ exponentially fast if:
\begin{equation}\label{eq:AWGN-good}
 \gamma_{\Lambda_{y}}\left(\tilde{\sigma}_{y}\right)
 =\frac{V\left(\Lambda_{y}\right)^{2 / n}}{\tilde{\sigma}_{y}^{2}}>2 \pi e,
\end{equation}
Additionally, since \textit{secrecy-goodness} of $\Lambda_w^{(n)}$ (with respect to $\Tilde{\sigma}_w=\frac{\sqrt{P}\sigma_w}{\sqrt{P^2+\sigma_w^2}}$), we have \cite{erez2004achieving,ling2014semantically}:
\begin{equation}\label{eq:Secrecy-good}
 \gamma_{\Lambda_{w}} \left(\tilde{\sigma}_{w}\right)
 = 
 \frac{V\left(\Lambda_{w}\right)^{2 / n}}{\tilde{\sigma}_{w}^{2}}<2 \pi.
\end{equation}

\noindent We have that strong secrecy rates $R$ satisfying
\begin{equation}\label{eq:R_2serv1}
R=\frac{1}{n} \log \frac{V\left(\Lambda_{w}\right)}{V\left(\Lambda_{y}\right)}
<
\frac{1}{2} \log \left(\frac{\frac{1}{2}+SNR_y}{1+SNR_w}\right)-\frac{1}{2}
\end{equation}
where $SNR_y = \frac{P}{\sigma_y^2}$,and $SNR_y = \frac{P}{\sigma_w^2}$.\\
Two extra conditions are required for the average power constraint and for the equivalent noise to be asymptotically Gaussian, that is
\begin{equation*}
V\left(\Lambda_{w}\right)^{2 / n}<2 \pi\frac{P}{\pi-1}, \ 
V\left(\Lambda_{w}\right)^{2 / n}<2 \pi\frac{P}{1+1 / \rho_{b}}
\end{equation*}
which together with \eqref{eq:AWGN-good} yields:
\begin{equation}\label{eq:R_2serv2}
\begin{split}
&R<\frac{1}{2} \log \left(\frac{\frac{1}{2}+\SNR_{y}}{\frac{\pi}{\pi-1 / e}}\right)-\frac{1}{2}
\\&
R<\frac{1}{2} \log \left(SNR_{y}\frac{\frac{1}{2}+ SNR_y}{1+ SNR_y}\right)-\frac{1}{2}
\end{split}
\end{equation}

Combining \eqref{eq:R_2serv1}-\eqref{eq:R_2serv2} together, we obtain Theorem \ref{th:the PIR rate for 2S public AWGN MISO channel}.

\begin{remark}
    Although the servers generate $r_1$ and $r_2$ independently, the user can decode the message.
\end{remark}

Next, we demonstrate that the \textit{privacy} requirement \eqref{equ-secrecy-privacy} for the $j$-th server is fulfilled.
Lemma \ref{lemma:theta,queries,messages} below shows that regardless of the codewords construction, the query and the messages reveal nothing about the index $\theta$.
\begin{lemma}\label{lemma:theta,queries,messages}
    \begin{equation}
        I(\theta;Q_g(\theta),W_1^M)=0, \quad g\in\{1,2\}
    \end{equation}
\end{lemma}
\begin{proof}
    This lemma follows since  $Q_g(\theta)$ is independent of $\theta$. 
    This can be verified from the probability function of $Q(\theta)$, that is $P(Q(\theta)_j=a|\theta = i)=P(Q(\theta)_j=a)=1/2$ for any $j\in\{1\dots,M)$.  
    Thus, $Q(\theta)$ is uniformly distributed over $\{0,1\}^M$ or $\{0,-1\}^M$.
    Additionally, $W_1^M$ is independent of $\theta$ and $Q_g(\theta)$.
\end{proof}
\begin{lemma}\label{lemma:strong_sec}
    \begin{equation}
    \lim_{n\rightarrow\infty}I(\lambda_1,\lambda_2;\mathbf{w}_j)=0, \quad  j\in\{1,2\}
    \end{equation}
\end{lemma}
\begin{proof}
\begin{equation*}
\begin{split}
     &I(\lambda_1,\lambda_2;\mathbf{w}_j) 
     \\
     &=I(\lambda_1;\mathbf{w}_j) + I(\lambda_2;\mathbf{w}_j\mid \lambda_1)
     \\
     &=H(\lambda_1)-  H(\lambda_1\mid \mathbf{w}_j )
     + H(\lambda_2\mid \lambda_1) - H(\lambda_2\mid \lambda_1,\mathbf{w}_j)  
     \\
     &\overset{(a)}{\leq}H(\lambda_1)-  H(\lambda_1\mid \mathbf{w}_j,\lambda_2)
     + H(\lambda_2) - H(\lambda_2\mid \lambda_1,\mathbf{w}_j)  
    \\
    &\overset{}{=}H(\lambda_1)-  H(\lambda_1\mid \lambda_1+n_{j,w},\lambda_2)
    \\& \hspace{2cm}
     + H(\lambda_2) - H(\lambda_2\mid \lambda_2+n_{j,w},\lambda_1)
     \\&
     \overset{(b)}{=}H(\lambda_1)-  H(\lambda_1\mid \lambda_1+n_{j,w})
     + H(\lambda_2) - H(\lambda_2\mid \lambda_2+n_{j,w})
     \\&
     \overset{}{=}I(\lambda_1;\mathbf{z}_1) + I(\lambda_2;\mathbf{z}_2)
     \overset{n\rightarrow\infty}{\longrightarrow}0
    \end{split}
\end{equation*}
    where (a) follows since conditioning reduces entropy and since $\lambda_1$ is independent with $\lambda_2$.
    (b) follows again since $\lambda_1$ independent with $\lambda_2$.
    The final step leverages the strong secrecy ensured by our scheme, specifically preventing the curious server from decoding $\lambda_i$ from the received signal.

Now, we ready to prove the \textit{Privacy} \eqref{equ-secrecy-privacy} requirement is fulfilled,
\begin{equation*}
\begin{split}   
&I(\theta;Q_j(\theta),W_1^M,\mathbf{w}_j)
\\&
=I(\theta;Q_j(\theta),W_1^M)+I(\theta;\mathbf{w}_j\mid Q_j(\theta),W_1^M)
\\&
\overset{(a)}{=} H(\mathbf{w}_j\mid Q_j(\theta),W_1^M)-H(\mathbf{w}_j\mid Q_j(\theta),W_1^M,\theta)
\end{split}
\end{equation*}
where (a) is by Lemma \ref{lemma:theta,queries,messages}.
Assume, without loss of generality, that $j=2$.
We have,
\begin{equation*}
    \begin{split}
    &
    H(\mathbf{w}_2\mid Q_2(\theta),W_1^M)-H(\mathbf{w}_2\mid Q_2(\theta),W_1^M,\theta)
    \\&
    \leq H(\mathbf{w}_2\mid Q_2(\theta),W_1^M)-H(\mathbf{w}_2\mid Q_2(\theta),W_1^M,\theta,\lambda_1,\lambda_2)      
    \\&
    \overset{(b)}{\leq} H(\mathbf{w}_2) - H(\mathbf{w}_2\mid\lambda_1,\lambda_2) \overset{n\rightarrow\infty}{\longrightarrow}0  
    \end{split}
\end{equation*}
(b) is due to the following Markov chain $\theta \leftrightarrow \left(Q_1(\theta),Q_2(\theta),W_1^M\right) \leftrightarrow \lambda_1,\lambda_2 \leftrightarrow \mathbf{w}_j$.
The last step is due to Lemma \ref{lemma:strong_sec}.
\end{proof}

\subsection{A Joint Scheme for Secure PIR Over Non-Fading AWGN MAC With $N>2$ Servers}
This section extends the previously discussed results to scenarios involving more than $N=2$ servers.
Specifically, over a transmission of $n$ symbols, the received signal at the user's antenna is given by:
\begin{equation}
\b{y}(i)=\sum_{k=1}^{N}\b{x}_k(i)+\mathbf{n}_y,
\end{equation}
where $\b{x}_k$ represents the signal transmitted by the $k$-th server, and $\mathbf{n}_y$ denotes the noise at the user.
Similarly, the received signal at the curious server antenna is described as:
\begin{equation}
\b{w}_j(i)=\sum_{k=1}^{N}\b{x}_k(i)+\mathbf{n}_{w},
\end{equation}
where $\mathbf{n}_w$ represents the noise at the malicious server.
\textcolor{black}{It is assumed that the signals received by all other servers are subject to noisier channels.
Consequently, ensuring security against the malicious server inherently provides protection against all other servers.}

As proposed in \cite{shmuel2021private}, the user generates two queries (such as \eqref{equ-queries structure PIR scheme}) and divides the $N$ servers randomly into two groups of $\floor{\frac{N}{2}}$ (in case $N$ is odd we can neglect one server).
Each group subsequently constructs its answer based on the query assigned to it.
\textcolor{black}{
The main distinction from the two-server scheme is lies in ensuring that each group generates the same random lattice point (coset representative). 
Intuitively, this alignment is crucial to our scheme, as the transmissions must sum coherently to enable successful decoding. To achieve this coordination, the user provides each group with a uniform binary vector, ensuring consistent generation of the lattice points within each group.
This vector is used to generate a sample from the relevant distribution \cite{cover2012elements}\footnote{Generating a sample from a discrete distribution is a well-studied topic \cite{cover2012elements}.
The expected number of fair bits required by an optimal algorithm to generate a random variable $X$ is less than $H(X) + 2$.}, specifically
$D_{\Lambda_w+\boldsymbol{\lambda}_i,\sqrt{P}}$ where $i\in\{1,2\}$.
}
Thus, the received input at the user and the malicious server is as follows:
\begin{equation}
\begin{split}
    \b{y}(i) = \floor{\frac{N}{2}}\b{x}_1(i)+\floor{\frac{N}{2}}\b{x}_2(i)+\mathbf{n}_y\\
    \b{w}(i) =\floor{\frac{N}{2}}\b{x}_1(i) + \floor{\frac{N}{2}}\b{x}_2(i)+\mathbf{n}_w
\end{split}
\end{equation}
Without the loss of generality, assume that the malicious server belongs to group $2$ and can perfectly subtract the answers corresponding to its group\footnote{\textcolor{black}{Note that this assumption is quite strong, as the malicious server would need to know the exact number of servers participating in the transmission — a nontrivial piece of information.}}. 
Its objective becomes decoding $\b{x}_1$ from the remaining received signals,
\begin{equation}
\begin{split}
    \b{z}_1(i) =\floor{\frac{N}{2}}\b{x}_1(i) +\mathbf{n}_w
\end{split}
\end{equation}
\noindent This case is equivalent to the two servers scenario where the transmission power at each antenna is $(\floor{\frac{N}{2}})^2P$.

The only thing we have to verify is the decoding stage.

\noindent\textit{Decode:} 
To decode $\mathbf{v}\overset{\Delta}{=}\boldsymbol{\phi}(W_i^{(1)})$, the user computes the following, 
\begin{equation*}
\hat{\b{v}}=[Q_{\Lambda_f}(\alpha'\b{y})]\ \text{ mod } \Lambda_w,
\end{equation*}   
where $\alpha'=\alpha\frac{1}{\floor{\frac{N}{2}}}$ and $\alpha$ will be given below.
Thus,
\begin{equation*}
    \begin{split}
        \hat{\b{v}}&=[Q_{\Lambda_f}(\alpha'\b{y})]\ \text{ mod } \Lambda_w
        \\&
        =[Q_{\Lambda_f}(\alpha(\b{x}_1+\b{x}_2+\frac{1}{\floor{\frac{N}{2}})}\b{n}_y)]\ \text{ mod } \Lambda_w
        \\&
        =[Q_{\Lambda_f}(\b{x}_1+\b{x}_2+(\alpha-1)(\b{x}_1+\b{x}_2)+\alpha\frac{1}{\floor{\frac{N}{2}})}\b{n}_y))]\ \text{ mod } \Lambda_w
        \\&
        [Q_{\Lambda_f}(\b{x}_1+\b{x}_2+\b{\Tilde{n}}_y))]\ \text{ mod } \Lambda_w
    \end{split}
\end{equation*}
\noindent where we denote the equivalent noise as ${\b{\Tilde{n}}_y\overset{\Delta}{=}(\alpha-1)(\b{x}_1+\b{x}_2)+\alpha'\b{n}_y}$ where $\frac{1}{n}\EX[\norm{\b{\Tilde{n}}_y}]=0$. 
Since $\mathbf{x}_{1}$ and $\mathbf{x}_{2}$ drawn independently  from $D_{\Lambda_w+\boldsymbol{\lambda}_i,\sqrt{P}}$, the second moment of ${\b{\Tilde{n}}_y}$ is given by
${\Tilde{\sigma}_y^{2}=\frac{1}{n}\EX\left[\norm{\mathbf{n}_{eq}^2}\right] =2(1-\alpha)^2P+\alpha^2 \floor{\frac{N}{2}}^{-2}\sigma_{y}^2}$ where we can optimize it on $\alpha$ getting $\alpha_{opt}=\frac{2P}{2P+\floor{\frac{N}{2}}^{-2}\sigma_{y}^2}$
and the resulting optimal second moment
is $\Tilde{\sigma}^{2}_{y,opt}=\frac{2P\floor{\frac{N}{2}}^{-2}\sigma_{y}^2}{2P+\floor{\frac{N}{2}}^{-2}\sigma_{y}^2}$.

Using \eqref{eq:QuantizationMod} we have,
\begin{equation*}
    \begin{split}
        &[Q_{\Lambda_f}(\b{x}_1+\b{x}_2+\b{\Tilde{n}}_y)]\ \text{ mod } \Lambda_w
        \\
       &  = [Q_{\Lambda_f}([\b{x}_1+\b{x}_2+\b{\Tilde{n}}_y]\text{ mod } \Lambda_w)]\ \text{ mod } \Lambda_w
       \\
       &   = [Q_{\Lambda_f}([\b{v}+\b{r}_1+\b{r}_2+\b{\Tilde{n}}_y]\text{ mod } \Lambda_w)]\ \text{ mod } \Lambda_w
       \\
       &\overset{(a)}{=} [Q_{\Lambda_f}([\b{v}+\b{\Tilde{n}}_y]\text{ mod } \Lambda_w)]\ \text{ mod } \Lambda_w
    \end{split}
\end{equation*}
\noindent (a) since $\b{r}_1\in\Lambda_w$ and $\b{r}_2\in\Lambda_w$ together with \eqref{eq:dist}.
Therefore
$$\hat{\b{v}}=[Q_{\Lambda_f}(\b{v}+\b{\Tilde{n}}_y)]\ \text{ mod } \Lambda_w $$
Finally, following the same steps as in the two servers,
as $n\rightarrow\infty$, we have that strong secrecy rate satisfies:
\begin{equation*}
    R<\frac{1}{2} \log^+ \left(\min \left\{\frac{\frac{1}{2}+SNR_{\tilde{y}}}{1+SNR_{\tilde{w}}}, SNR_{y}'\frac{\frac{1}{2}+ SNR_{\tilde{y}}}{1+ SNR_{\tilde{y}}}\right\}\right)-\frac{1}{2}
\end{equation*}
where $SNR_{\tilde{y}}=\floor{\frac{N}{2}}^{2}\frac{P}{\sigma^2_y}$ and $SNR_{\tilde{w}}=\floor{\frac{N}{2}}^{2}\frac{P}{\sigma^2_w}$.

\subsection{A Joint Scheme for PIR Over Public Block-Fading AWGN MAC with $N>2$ Servers}
In this section, we consider a block-fading AWGN MAC.
As discussed in the previous section, the transmissions here are affected not only by noise but also by attenuation due to the channel's fading characteristics \eqref{eq:MAC_output_model}.
The core concept of the proposed scheme involves dividing the servers into two groups and exploiting the channel characteristics to cancel the messages between them, leaving only the desired message at the receiver. The primary challenge is ensuring that the responses from both groups are properly aligned. To address this, the group with the higher channel gain reduces its transmission power to match the overall power of the two groups, facilitating alignment at the receiver.

\subsubsection{Coding Scheme}
The construction of the codebook is the same as in the two-server scheme.

\noindent \textit{Query:} The user, which is interested in the message $W_i$, generates a random vector $\b{b}$ of length $M$ such that each entry is either $1$ or $0$, independently and with equal probability. Then, the user divides the databases into two non-intersecting subsets, denoted as $\mathcal{S}_1$ and $\mathcal{S}_2$, for which he sends the query $Q_1(i)$ to each member in $\mathcal{S}_1$ and $Q_2(i)$ to each member in $\mathcal{S}_2$. 
The queries are given as follows
\begin{equation}\label{equ-queries structure PIR scheme - fading}
\begin{aligned}
&Q_{1}(i)=\b{b}, \ \ Q_{2}(i)=-\b{b}-\b{e}_i, \ \text{ if } b_i=0\\
&Q_{1}(i)=\b{b}, \ \ Q_{2}(i)=-\b{b}+\b{e}_i, \ \text{ if } b_i=1
\end{aligned}
\end{equation}

\noindent 
\textit{Answers:} Upon receiving the queries, the servers construct their responses by computing linear combinations of the messages, where the combining coefficients are determined by the entries in $Q_j(i)$. 
That is,
\begin{equation}\label{equ-servers answer formation-fading}
\b{A}_k=\sum_{m=1}^M Q_{k}(m)W_m,
\end{equation}
where $Q_{k}(m)$ is the $m$th entry of the vector $Q_{k}$.
We note that $A_k=(a_k^{(1)},\dots,a_k^{(L)})\in \mathbb{F}_p^L$, and the scheme is focused on the transmission of a single symbol $a_k^{(m)}$, where $1\leq m \leq L$, from each answer.
To construct the entire message, the servers must iterate this process across all $L$ symbols.

Without any loss of generality, assume $m=1$, i.e., the servers wish to transmit the first symbol from each answer, that is  $a_1^{(1)} $ and $a_2^{(1)}$.
Note that $\b{A}_1+\b{A}_{2}$ is equal to either $W_i$, or $-W_i$. 
This depends on the sign of $b_i$, which is known to the user.
In the same way $a_1^{(m)}+a_2^{(m)} = \pm W_i^{(m)}$.
To encode the symbols, each server maps the symbols to the codebook, namely, $\boldsymbol{\lambda}_1 = \boldsymbol{\phi}(a_1^{(1)})$ and  $\boldsymbol{\lambda}_2 = \boldsymbol{\phi}(a_2^{(1)})$.
Then, each server samples randomly and independently a lattice point from the relevant coset according to $D_{\Lambda_w+\boldsymbol{\lambda}_i,\sqrt{P}}$ where $i\in\{1,2\}$ and $P$ is the transmission power.
Equivalently, the $i$-th server which belong to $\mathcal{S}_g$, $g\in\{1,2\}$ transmits:
\begin{equation}
\begin{split}
    &\mathbf{x_g} = \boldsymbol{\lambda}_g + \mathbf{r}_g, \ \ \mathbf{r}_g\sim D_{\Lambda_w,\sqrt{P},-\boldsymbol{\lambda}_g}
\end{split}
\end{equation}
\textcolor{black}{The scheme requires that each group will generate the same sample. 
As in the non-fading scheme, the user is required to provide each group with a uniform binary vector. 
This vector will be used to generate a sample from the relevant distribution.}

Over a transmission of $n$ symbols, the received signal at the user's antenna and the malicious server (we consider the one with the best channel gain), respectively, is given by:
\begin{equation}
\b{y}(i)=\sum_{k\in\mathcal{S}_1}|h_k|\b{x}_1(i)+\sum_{k\in\mathcal{S}_2}|h_k|\b{x}_2(i)+\mathbf{n}_y 
\end{equation}
\begin{equation}
\b{w}(i)=\sum_{k\in\mathcal{S}_1}g_k\b{x}_1(i)+\sum_{k\in\mathcal{S}_2}g_k\b{x}_2(i)+\mathbf{n}_w,
\end{equation}
where $\b{x}_k$ represents the signal transmitted by the $k$-th server, and $\mathbf{n}_y\sim \mathcal{N}(0,\sigma_y^2\b{I}^{n\times n})$  and $\mathbf{n}_w\sim \mathcal{N}(0,\sigma_w^2\b{I}^{n\times n})$ are an i.i.d., Gaussian noise, and $h_{k} \sim \mathcal{N}(0,1)$, $g_{k} \sim \mathcal{N}(0,1)$  are the real channel coefficients.
Assuming full CSI,  each server can adjust its transmission by multiplying its answer by $-1$ if necessary to ensure that the responses sum up according to the absolute values of the fading coefficients.

Assume, without loss of generality,  $\Tilde{h}_1\leq\Tilde{h}_2$. 
Thus, in order to align the transmission, the servers that belong to $\mathcal{S}_2$ will reduce the transmission and transmit the following:
\begin{equation}
\mathbf{x_2}' = \frac{\Tilde{h}_1}{\Tilde{h}_2}\mathbf{x_2}
\end{equation}
where $\Tilde{h}_m = \sum_{k\in\mathcal{S}_m} |h_k|$ and $\Tilde{g}_m = \sum_{k\in\mathcal{S}_m} g_k$ (notice that only $h_k$ coefficients are with absolute value), which results in
\begin{equation}
\begin{split}
&\b{y}(i)=\tilde{h}_1\b{x}_1(i)+\tilde{h}_2\b{x}'_2(i)+\mathbf{n}_y
=\tilde{h}_1(\b{x}_1(i)+\b{x}_2(i))+\mathbf{n}_y 
\\&
\b{w}(i)=\tilde{g}_1\b{x}_1(i)+\tilde{g}_2 \b{x}'_2(i)+\mathbf{n}_w=\tilde{g}_1\b{x}_1(i)+\tilde{g}_2\frac{\tilde{h}_1}{\tilde{h}_2} \b{x}_2(i)+\mathbf{n}_w 
\end{split}
\end{equation}
Note that the user may add additional information to the query, informing the database which group the database belongs to and the factor to be multiplied before transmission.

\textit{Decode:} 
To decode $\mathbf{v}\overset{\Delta}{=}\boldsymbol{\phi}(W_i^{(1)})$, the user computes the following, 
\begin{equation*}
\hat{\b{v}}=[Q_{\Lambda_f}(\alpha\frac{1}{\tilde{h}_1}\b{y})]\ \text{ mod } \Lambda_w,
\end{equation*}   
where $\alpha$ where $0\leq\alpha\leq1$ will be determined later.
The above reduces to the Modulo-Lattice Additive Noise (MLAN) channel (\cite{erez2004achieving}) as follows,
\begin{equation*}
    \begin{split}
        \hat{\b{v}}&=[Q_{\Lambda_f}(\alpha\frac{1}{\tilde{h}_1}\b{y})]\ \text{ mod } \Lambda_w
        \\&
        =[Q_{\Lambda_f}(\alpha(\b{x}_1+\b{x}_2+\frac{1}{\tilde{h}_1}\b{n}_y))]\ \text{ mod } \Lambda_w
        \\&
        =[Q_{\Lambda_f}(\b{x}_1+\b{x}_2+(\alpha-1)(\b{x}_1+\b{x}_2)+\alpha\frac{1}{\tilde{h}_1}\b{n}_y))]\ \text{ mod } \Lambda_w
        \\&
        =[Q_{\Lambda_f}(\b{x}_1+\b{x}_2+\b{\Tilde{n}}_y))]\ \text{ mod } \Lambda_w
    \end{split}
\end{equation*}
\noindent where we denote the equivalent noise as ${\b{\Tilde{n}}_y\overset{\Delta}{=}(\alpha-1)(\b{x}_1+\b{x}_2)+\alpha\frac{1}{\tilde{h}_1}\b{n}_y}$ where $\frac{1}{n}\EX[\norm{\b{\Tilde{n}}_y}]=0$. 
Since $\mathbf{x}_{1}$ and $\mathbf{x}_{2}$ drawn independently  from $D_{\Lambda_w+\boldsymbol{\lambda}_i,\sqrt{P}}$, the second moment of ${\b{\Tilde{n}}_y}$ is given by
${\Tilde{\sigma}_y^{2}=\frac{1}{n}\EX\left[\norm{\mathbf{n}_{eq}^2}\right] =2(1-\alpha)^2P+\alpha^2\frac{1}{\tilde{h}_1^2} \sigma_{y}^2}$ where we can optimize it on $\alpha$ getting $\alpha_{opt}=\frac{2P}{2P+\frac{1}{\tilde{h}_1^2}\sigma_{y}^2}$
and the resulting optimal second moment
is $\Tilde{\sigma}^{2}_{y,opt}=\frac{2P\frac{1}{\tilde{h}_1^2}\sigma_{y}^2}{2P+\frac{1}{\tilde{h}_1^2}\sigma_{y}^2}$.

Using \eqref{eq:QuantizationMod} we have,
\begin{equation*}
    \begin{split}
        &[Q_{\Lambda_f}(\b{x}_1+\b{x}_2+\b{\Tilde{n}}_y)]\ \text{ mod } \Lambda_w 
        \\
       &  = [Q_{\Lambda_f}([\b{x}_1+\b{x}_2+\b{\Tilde{n}}_y]\text{ mod } \Lambda_w)]\ \text{ mod } \Lambda_w 
       \\
       &   = [Q_{\Lambda_f}([\b{v}+\b{r}_1+\b{r}_2+\b{\Tilde{n}}_y]\text{ mod } \Lambda_w)]\ \text{ mod } \Lambda_w 
       \\
       &\overset{(a)}{=} [Q_{\Lambda_f}([\b{v}+\b{\Tilde{n}}_y]\text{ mod } \Lambda_w)]\ \text{ mod } \Lambda_w 
    \end{split}
\end{equation*}
\noindent (a) since $\b{r}_1\in\Lambda_w$ and $\b{r}_2\in\Lambda_w$ together with \eqref{eq:dist}.
Therefore
$$\hat{\b{v}}=[Q_{\Lambda_f}(\b{v}+\b{\Tilde{n}}_y)]\ \text{ mod } \Lambda_w $$
Thus, from the \textit{AWGN-goodness} of $\Lambda_f$, $P_e$ tends to $0$ exponentially fast if 
\begin{equation}\label{eq:AWGN-good1}
 \gamma_{\Lambda_{y}}\left(\tilde{\sigma}_{y}\right)
 =\frac{V\left(\Lambda_{y}\right)^{2 / n}}{\tilde{\sigma}_{y}^{2}}>2 \pi e,
\end{equation}
On the other hand, we assume the worst-case scenario for the curious server: that is, the server has full CSI for the legitimate channels and the inter-server channels ($h_k$ and $g_k$ for all $k$). 
Additionally, we assume that the curious server knows which servers are participating in the transmission—a nontrivial assumption given that the servers do not communicate with each other.
Hence, the curious server can cancel part of the messages according to its group.
Thus, the curious server observes one of the following,
$\b{z}_1(i)=\tilde{g}_1\b{x}_1(i)+\mathbf{n}_w $ or $\b{z}_2(i)=\tilde{g}_2\frac{\tilde{h}_1}{\tilde{h}_2} \b{x}_2(i)+\mathbf{n}_w $.
We will take care of both cases. 
For the first case we choose $\Lambda_w^{(n)}$ with respect to $$\tilde{\sigma}_{w_1}=\frac{\sqrt{P}\sigma_w\frac{1}{\tilde{g}_1}}{\sqrt{P^2+\sigma_w^2\frac{1}{\tilde{g}_1^2}}},$$
and for the second, we choose
$\Lambda_w^{(n)}$ with respect to $$\tilde{\sigma}_{w_2}=\frac{\sqrt{P}\sigma_w\frac{1}{\tilde{g}_2}\frac{\tilde{h}_2}{\tilde{h}_1}}{\sqrt{P^2+\sigma_w^2\left(\frac{1}{\tilde{g}_2}\frac{\tilde{h}_2}{\tilde{h}_1}\right)^2 }}.$$

Let us define $SNR_{\tilde{y}} = \frac{\tilde{h}_1^2 P}{\sigma_{y}^2}$, 
$SNR_{\tilde{w}_1} = \frac{\tilde{g}_1^2P}{\sigma_{w}^2} $, and 
$SNR_{\tilde{w}_2} = \frac{\left(\tilde{g}_2\frac{\tilde{h}_1}{\tilde{h}_2}\right)^2 P}{\sigma_{w}^2} $.

 Thus, from \textit{secrecy-goodness} of $\Lambda_w^{(n)}$ (with respect to $\tilde{\sigma}_{w_1}$ and $\tilde{\sigma}_{w_2}$), \cite{erez2004achieving,ling2014semantically}, we obtain Theorem \ref{th:the PIR rate for block fading public AWGN MISO channel}.

We note that the suggested scheme for the block-fading AWGN channel does not impair the \textit{Privacy} \eqref{equ-secrecy-privacy} requirement, even if the channel coefficients ($\mathbf{h}$ and $\mathbf{g}$)  are globally known.
Specifically, we have
\begin{equation*}
\begin{split}   
&I(\theta;Q_j(\theta),W_1^M,\mathbf{w},\mathbf{h},\mathbf{g})
\\
&=I(\theta;\mathbf{h},\mathbf{g})
+
I(\theta;Q_j(\theta),W_1^M,\mathbf{w}\mid\mathbf{h},\mathbf{g})
\\&
=I(\theta;Q_j(\theta),W_1^M\mid\mathbf{h},\mathbf{g})
+I(\theta;\mathbf{w}\mid Q_j(\theta),W_1^M,\mathbf{h},\mathbf{g})
\\&
\overset{(a)}{=}I(\theta;Q_j(\theta),W_1^M)
+I(\theta;\mathbf{w}\mid Q_j(\theta),W_1^M,\mathbf{h},\mathbf{g})
\\&
\overset{(b)}{=} H(\mathbf{w}\mid Q_j(\theta),W_1^M,\mathbf{h},\mathbf{g})
-H(\mathbf{w}\mid Q_j(\theta),W_1^M,\theta,\mathbf{h},\mathbf{g})
\\&
\overset{(c)}{\leq} H(\mathbf{w}\mid\mathbf{h},\mathbf{g})
-H(\mathbf{w}\mid Q_2(\theta),W_1^M,\theta,\mathbf{h},\mathbf{g})
\\&
\overset{}{=}I(\mathbf{w}; Q_2(\theta),W_1^M,\theta\mid\mathbf{h},\mathbf{g})
\\& \overset{}{\leq}
I(\mathbf{w}; Q_2(\theta),W_1^M,\theta,\lambda_1,\lambda_2\mid \mathbf{h},\mathbf{g})
\\& \overset{(d)}{=}
I(\mathbf{w};\lambda_1,\lambda_2\mid \mathbf{h},\mathbf{g})
\\&=
I(\mathbf{w};\lambda_1\mid \mathbf{h},\mathbf{g})
+
I(\mathbf{w};\lambda_2\mid \mathbf{h},\mathbf{g},\lambda_1)
\\&
=H(\lambda_1\mid\mathbf{h},\mathbf{g})
-H(\lambda_1\mid\mathbf{h},\mathbf{g},\mathbf{w})
\\& \quad + H(\lambda_2\mid\mathbf{h},\mathbf{g},\lambda_1)
-
H(\lambda_2\mid\mathbf{h},\mathbf{g},\lambda_1,\mathbf{w})
\\& 
\leq H(\lambda_1)
-H(\lambda_1\mid\mathbf{h},\mathbf{g},\mathbf{w},\mathbf{x}_2)
\\& \quad + H(\lambda_2)
-
H(\lambda_2\mid\mathbf{h},\mathbf{g},\lambda_1,\mathbf{w},\mathbf{x}_1)
\\&
\overset{(e)}{=}
H(\lambda_1)
-H(\lambda_1\mid\mathbf{z}_1) + H(\lambda_2)
-H(\lambda_2\mid\lambda_1,\mathbf{z}_2)
\\&
= I(\lambda_1;\mathbf{z}_1) + I(\lambda_2;\mathbf{z}_2)\overset{n\rightarrow\infty}{\longrightarrow}0
\end{split}
\end{equation*}
(a) since ($\mathbf{h},\mathbf{g}$) is independent with ($\theta,Q_j(\theta),W_1^M$) and (b) is due to Lemma \ref{lemma:theta,queries,messages}.
(c) assume without loss of generality that $j=2$.
(d) is due to the Markov chain $\theta\leftrightarrow\left(Q_1(\theta),Q_2(\theta),W_1^M\right)\leftrightarrow\left(\lambda_1,\lambda_2,\mathbf{h},\mathbf{g}\right)\leftrightarrow \mathbf{w}$.
(e) is due to the Markov chains
$\lambda_1\leftrightarrow\mathbf{x}_1\leftrightarrow\mathbf{z}_1\leftrightarrow(\mathbf{h},\mathbf{g},\mathbf{x}_2,\mathbf{w})$, and $\lambda_2\leftrightarrow\mathbf{x}_2\leftrightarrow\mathbf{z}_2\leftrightarrow(\mathbf{h},\mathbf{g},\mathbf{x}_1,\mathbf{w})$.
The last step follows from the strong secrecy implied by our scheme.

To analyze the achievable PIR rate in Theorem \ref{th:the PIR rate for block fading public AWGN MISO channel} we note that the user may choose $\mathcal{S}_1$ and $\mathcal{S}_2$ to maximize:

\begin{equation}\label{eq:maximization}
\max_{\mathcal{S}_1, \mathcal{S}_2}\left\{ \log \left(\min_i \left\{\frac{\frac{1}{2}+SNR_{\tilde{y}}}{1+SNR_{\tilde{w}_i}}, SNR_{\tilde{y}}\frac{\frac{1}{2}+ SNR_{\tilde{y}}}{1+ SNR_{\tilde{y}}}\right\}\right)\right\}
\end{equation}

\begin{figure}
    \centering
    \includegraphics[width=0.9\linewidth]{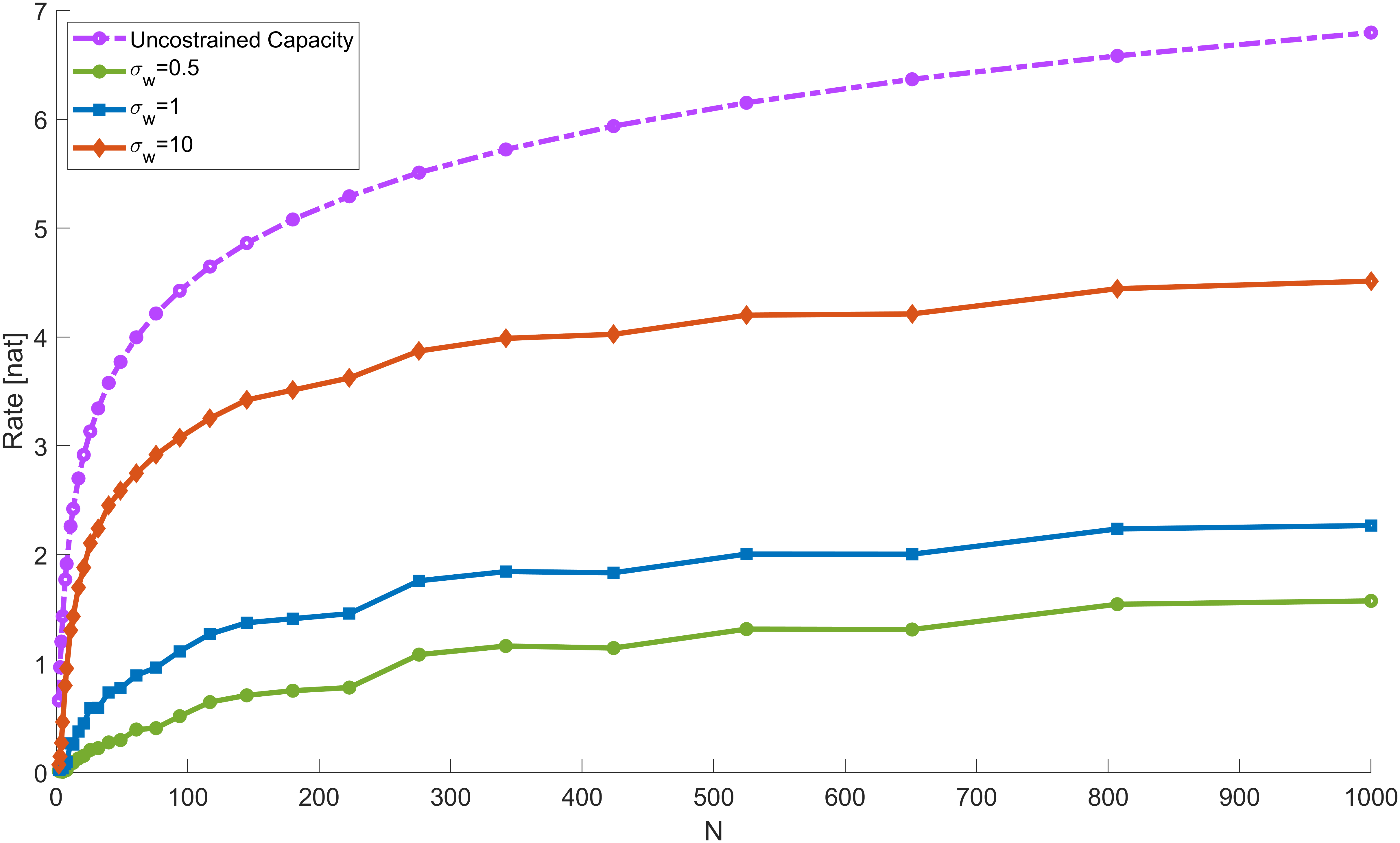}
    \caption{The average PIR rate as a function of the number of servers $N$. 
    The transmit power is set to $P=1$ and the different colors represent varying values of $\sigma_w$ with $\sigma_y$ fixed at $1$. }
    \label{fig:fading_PIR_rate_vs_N}
\end{figure}

Finding the optimal solution to this problem is closely related to the subset sum problem, known as NP-hard. 
Nevertheless, greedy algorithms can be employed to provide suboptimal solutions within a reasonable time frame.
In Figure \ref{fig:fading_PIR_rate_vs_N} we illustrate the average achievable PIR rate as a function of the number of servers, assuming $\sigma_y=\sigma_w=1$.
Interestingly, a positive PIR rate can still be achieved even with equal noise variances or even greater variance at the user antenna. 
This is primarily because the channel coefficients between the servers and the user are effectively combined due to the knowledge of CSI. In contrast, the channel coefficients observed by the curious server may cancel each other, leading to a reduced SNR for the curious server.

\section{Conclusion}
We proposed a joint PIR-coding scheme for an AWGN MAC channel, where the public nature of the responses could potentially lead to privacy leakage if one of the servers is malicious. 
In the block-fading case, we demonstrated that a positive PIR rate is achievable, even when the channel to the malicious server is better than the channel to the user.

It is worth noting that a loss of $1/2$ nat occurs due to using the wiretap coding scheme from \cite{ling2014semantically}.
This loss could potentially be avoided by employing the scheme proposed in \cite{liu2018achieving}, which utilizes a polar lattice code to eliminate the $\frac{1}{2}$-nat gap from the secrecy capacity of the Gaussian Wiretap Channel (GWC). 
However, we chose to use the scheme from \cite{ling2014semantically} for the sake of convenience and clarity in presentation, as this choice does not affect the scaling laws of our results.


\bibliographystyle{IEEEtran}
\bibliography{references}

\begin{thebibliography}{10}
\providecommand{\url}[1]{#1}
\csname url@samestyle\endcsname
\providecommand{\newblock}{\relax}
\providecommand{\bibinfo}[2]{#2}
\providecommand{\BIBentrySTDinterwordspacing}{\spaceskip=0pt\relax}
\providecommand{\BIBentryALTinterwordstretchfactor}{4}
\providecommand{\BIBentryALTinterwordspacing}{\spaceskip=\fontdimen2\font plus
\BIBentryALTinterwordstretchfactor\fontdimen3\font minus \fontdimen4\font\relax}
\providecommand{\BIBforeignlanguage}[2]{{%
\expandafter\ifx\csname l@#1\endcsname\relax
\typeout{** WARNING: IEEEtran.bst: No hyphenation pattern has been}%
\typeout{** loaded for the language `#1'. Using the pattern for}%
\typeout{** the default language instead.}%
\else
\language=\csname l@#1\endcsname
\fi
#2}}
\providecommand{\BIBdecl}{\relax}
\BIBdecl

\bibitem{sun2017capacity}
H.~Sun and S.~A. Jafar, ``The capacity of private information retrieval,'' \emph{IEEE Transactions on Information Theory}, vol.~63, no.~7, pp. 4075--4088, 2017.

\bibitem{sun2017capacityColluding}
------, ``The capacity of robust private information retrieval with colluding databases,'' \emph{IEEE Transactions on Information Theory}, vol.~64, no.~4, pp. 2361--2370, 2017.

\bibitem{banawan2018capacity}
K.~Banawan and S.~Ulukus, ``The capacity of private information retrieval from byzantine and colluding databases,'' \emph{IEEE Transactions on Information Theory}, vol.~65, no.~2, pp. 1206--1219, 2018.

\bibitem{banawan2019noisy}
------, ``Noisy private information retrieval: On separability of channel coding and information retrieval,'' \emph{IEEE Transactions on Information Theory}, 2019.

\bibitem{shmuel2021private}
O.~Shmuel and A.~Cohen, ``Private information retrieval over gaussian mac,'' \emph{IEEE Transactions on Information Theory}, vol.~67, no.~8, pp. 5404--5419, 2021.

\bibitem{orelimcorrection24}
O.~Elimelech, O.~Shmuel, and A.~Cohen, ``Corrections to “private information retrieval over gaussian mac”,'' \emph{IEEE Transactions on Information Theory}, vol.~70, no.~10, pp. 7521--7524, 2024.

\bibitem{orelimEfficientISIT}
O.~Elimelech and A.~Cohen, ``An efficient, high-rate scheme for private information retrieval over the gaussian mac,'' in \emph{2024 IEEE International Symposium on Information Theory (ISIT)}, 2024, pp. 3672--3677.

\bibitem{wang2018capacity}
Q.~Wang, H.~Sun, and M.~Skoglund, ``The capacity of private information retrieval with eavesdroppers,'' \emph{IEEE Transactions on Information Theory}, vol.~65, no.~5, pp. 3198--3214, 2018.

\bibitem{cover2012elements}
T.~M. Cover and J.~A. Thomas, \emph{Elements of information theory}.\hskip 1em plus 0.5em minus 0.4em\relax John Wiley \& Sons, 2012.

\bibitem{erez2004achieving}
U.~Erez and R.~Zamir, ``Achieving 1/2 log (1+ snr) on the awgn channel with lattice encoding and decoding,'' \emph{IEEE Transactions on Information Theory}, vol.~50, no.~10, pp. 2293--2314, 2004.

\bibitem{ling2014achieving}
C.~Ling and J.-C. Belfiore, ``Achieving awgn channel capacity with lattice gaussian coding,'' \emph{IEEE Transactions on Information Theory}, vol.~60, no.~10, pp. 5918--5929, 2014.

\bibitem{zamir2014lattice}
R.~Zamir, \emph{Lattice Coding for Signals and Networks: A Structured Coding Approach to Quantization, Modulation, and Multiuser Information Theory}.\hskip 1em plus 0.5em minus 0.4em\relax Cambridge University Press, 2014.

\bibitem{ling2014semantically}
C.~Ling, L.~Luzzi, J.-C. Belfiore, and D.~Stehl{\'e}, ``Semantically secure lattice codes for the gaussian wiretap channel,'' \emph{IEEE Transactions on Information Theory}, vol.~60, no.~10, pp. 6399--6416, 2014.

\bibitem{nazer2011compute}
B.~Nazer and M.~Gastpar, ``Compute-and-forward: Harnessing interference through structured codes,'' \emph{IEEE Transactions on Information Theory}, vol.~57, no.~10, pp. 6463--6486, 2011.

\bibitem{liu2018achieving}
L.~Liu, Y.~Yan, and C.~Ling, ``Achieving secrecy capacity of the gaussian wiretap channel with polar lattices,'' \emph{IEEE Transactions on Information Theory}, vol.~64, no.~3, pp. 1647--1665, 2018.

\end{thebibliography}
\end{document}